\documentclass{confSt}
 
\begin{document}
\def\bG{\textbf{\textit{\textsf{G}}}}
\def\calv{\mathcal{V}}
\def\Pra{\textbf{Pr}_a}
\def\Prd{\textbf{Pr}_d}
\def\EP{\textbf{ExP}}
\def\N{\textsf{\textbf{N}}}
\def\coSC{\textsf{coSC}}
\def\SC{\textsf{SC}}
\def\SCALL{\textsf{SC\_ALL}}
\def\Va{V_a}
\def\Vns{V_{ns}}
\def\Vd{V_d}
\def\atk{\textsf{atk}}
\def\ATK{\textsf{ATK}}
\def\DEF{\textsf{DEF}}
\def\Vsrc{V_{src}}
\def\Vld{V_{ld}}
\def\disc{\textsf{disc}_{\Vld,\Vsrc}}
\def\near{\mathcal{N}}
\def\Vsing{V_{sing}}

\def\cij{c_{ij}}
\def\F{\textbf{\textsf{F}}}
\def\Fst{\F^{*}}
\def\kd{k_d}
\def\ka{k_a}
\def\sp{\textit{load}_{G}}
\def\spZ{\textit{load}_{G_0}}
\def\spGp{\textit{load}_{G'}}

\def\spVE{\textit{load}_{(V,E)}}

\def\vgen{V_P^{(gen)}}
\def\vld{V_P^{(ld)}}
\def\vjct{V_P^{(jct)}}
\def\spP{sp^{\vgen,\vld}}
\def\cpc{cap^{(pwr)}}
\def\gpNs{G_P^{(ns)}}
\def\failP{fail^{(pwr)}}
\def\failR{fail^{(rtr)}}
\def\removeR{rem}
\def\remove{rem}
\def\S{\textbf{S}}
\def\T{\textbf{T}}

\def\vbrdr{V_R^{(brdr)}}
\def\vtgt{V_R^{(tgt)}}
\def\vsav{V_R^{(sav)}}

\def\cpcR{cap^{(rtr)}}
\def\spR{sp}
\def\src{src}
\def\inactive{inactive_{V_R}}

\def\Spwr{\textbf{S}^{(pwr)}}
\def\Tpwr{\textbf{T}^{(pwr)}}
\newtheorem{proposition}{Proposition}[section]
\newtheorem{fact}{Fact}[section]
\newtheorem{definition}{Definition}[section]

\title{Power Grid Defense Against Malicious Cascading Failure}

\numberofauthors{3}

\author{
Paulo Shakarian \\       
\affaddr{Dept. EECS and}\\
\affaddr{Network Science Center}\\
       \affaddr{U.S. Military Academy}\\
       \affaddr{West Point, NY, 10996}\\
       \email{paulo[at]shakarian.net}
\alignauthor
Hansheng Lei \\
\affaddr{Dept. EECS and}\\
\affaddr{Network Science Center}\\
       \affaddr{U.S. Military Academy}\\
       \affaddr{West Point, NY, 10996}\\
       \email{hansheng.lei[at]usma.edu}
\alignauthor 
Roy Lindelauf \\
       \affaddr{Netherlands Defence Academy}\\
       \affaddr{Faculty of Military Science}\\
       \affaddr{Military Operational Art and Science}\\
       \email{rha.lindelauf.01[at]nlda.nl}
}

\maketitle

\begin{abstract}
An adversary looking to disrupt a power grid may look to target certain substations and sources of power generation to initiate a cascading failure that maximizes the number of customers without electricity.  This is particularly an important concern when the enemy has the capability to launch cyber-attacks as practical concerns (i.e. avoiding disruption of service, presence of legacy systems, etc.) may hinder security.  Hence, a defender can harden the security posture at certain power stations but may lack the time and resources to do this for the entire power grid.  We model a power grid as a graph and introduce the cascading failure game in which both the defender and attacker choose a subset of power stations such as to minimize (maximize) the number of consumers having access to producers of power. We formalize problems for identifying both mixed and deterministic strategies for both players, prove complexity results under a variety of different scenarios, identify tractable cases, and develop algorithms for these problems.  We also perform an experimental evaluation of the model and game on a real-world power grid network.  Empirically, we noted that the game favors the attacker as he benefits more from increased resources than the defender.  Further, the minimax defense produces roughly the same expected payoff as an easy-to-compute deterministic load based (DLB) defense when played against a minimax attack strategy.  However, DLB performs more poorly than minimax defense when faced with the attacker's best response to DLB.  This is likely due to the presence of low-load yet high-payoff nodes, which we also found in our empirical analysis.
\end{abstract}

\category{I.2.11}{Artificial Intelligence}{Distributed Artificial Intelligence}

\terms{Algorithms} {Security}

\keywords{power grid defense, game theory, complex networks}
\section{Introduction}
Rapid cascading failure in a power grid caused by a succession of overloading lines can lead to very large outages, as observed in the United States in 2003~\cite{ohFailure}.  Studies on cascading failure~\cite{buld10,crucitti04,motter02} have illustrated that such a failure can be initiated with only a small number of initial node failures.  Further, power grid infrastructure is often particularly vulnerable with respect to cyber-security due to a variety of issues, including the use of legacy and proprietary computer hardware and software~\cite{wei}.

In this paper, we extend the work on cascading failure models to a two-player game where an attacker attempts to create a cascade that maximizes the number of customers without power while the defender defends key nodes to avoid a major outage.  In Section~\ref{trSec}, we introduce an extension to the failure model of \cite{crucitti04} to not only consider the attacker and defender, but also the different types of nodes in the power grid (i.e. power generation vs. power consumers).  In Section~\ref{ccSec}, we explore the computational complexity of finding deterministic best-response strategies for the attacker and defender under several different scenarios depending on the relative number of resources each player has and whether the opponent has a deterministic or mixed strategy.  Here we found that, in general, these problems are NP-hard, though we do identify some tractable cases.  In Section~\ref{analySec}, we explore heuristic algorithms for finding determinsitic ``best responses'' as well as minimax mixed strategies.  We introduce a ``high-load'' strategy for defense (based on the observations of~\cite{crucitti04}), greedy heuristics for deterministic strategies, and a double-oracle approach based on \cite{mcmahan03} for finding a mixed strategy.  In Section~\ref{expSec} we perform experiments on a real-world dataset of a power grid~\cite{rosato08} and find that this game seems to favor the attacker as he benefits more from increased resources than the defender.  Further, our experiments revealed that the minimax defense produces roughly the same expected payoff as an easy-to-compute deterministic load based (DLB) defense when played against a minimax attack strategy, though the load based defense does more poorly than minimax when faced with the attacker's best response to DLB.  This is likely due to the presence of low-load yet high-payoff nodes, which we also found in our empirical analysis of the model.  Finally, related work is discussed in Section~\ref{rwSec}.

\pagebreak
\section{Technical Preliminaries}
\label{trSec}
Consider a power-grid network modeled as an undirected graph $G=(V,E)$.  Let $\Vsrc, \Vld \subseteq V$ be source (producers of power) and load (consumers of power) on the network.  We shall use the notation $\disc(G)$ to denote the number of nodes in $\Vld$ which are not connected to any node in $\Vsrc$ in graph $G$.  Let $\bG$ be the set of all subgraphs of $G$.  For a given node $i$, let $\near_G(i)$ be the set of nodes in $\Vsrc-\{i\}$ that are closest to that node (based on path length in $G$).  From this, we define edge load (similar to the idea of edge betweenness~\cite{wasserman1994social}).

\begin{definition}[Edge Load]
Given edge $ij\in E$, the \textbf{edge load}, $\sp(ij)$ is defined as follows:
\begin{equation*}
\sp(ij) = \sum_{t \in \Vld}\sum_{s \in \near_G(t)}\frac{\sigma_G(s,t | ij)}{|\near_G(t)|\,\sigma_G(s,t)},
\end{equation*}
where $\sigma_G(s,t)$ is the number of shortest paths between $s,t \in V$ and $\sigma_G(s,t | ij)$ is the subset of these paths that pass through edge $ij\in E$. 
\end{definition}

Starting from initial network $G_0=(V_0,E_0)$ we use $c_{ij}$ to denote the capacity edge $ij\in E_0$.  In a real-world setting, we would expect to have this information.  However, in this paper, we use the following proxy (similar to \cite{crucitti04}).
\begin{equation*}
c_{ij}(G_0)=(1+\alpha){\spZ}(ij) 
\end{equation*}
where $\alpha$ is a non-negative real that specifies the excess capacity available on that line. We shall refer to $\alpha$ as the \textit{capacity margin}.  We assume that an edge $ij \in E$ fails in $G=(V,E)$, with $E \subset E_0$, if $\sp(ij) > c_{ij}(G_0)$. Once nodes (and adjacent edges) in $V_0$ are removed from $G_0$, this results in a change of shortest paths between sources and loads, hence more edges will potentially fail. This cascading power failure is modeled by a ``failure'' operator denoted with $\F$ (based on the failure model of \cite{crucitti04} - though we note that our model is a new contribution due to the consideration of source and load nodes) that maps networks to networks.  We define it as follows.

\begin{definition}[Failure Operator]
The \textbf{failure operator}, $\F: \bG \rightarrow \bG$, is defined as follows:
\begin{equation*}
\F((V,E)) = (V, \{ij \in E | \spVE(ij) \leq \cij(G_0) \})
\end{equation*}
\end{definition}

Intuitively, one application of the failure operator removes all edges that have exceeded their maximum capacity.  We can define multiple applications of this operator as follows:

\begin{equation*}
 \F^{i}(G) = \begin{cases} G & \text{if $i=0$}  \\ \F(\F^{i-1}(G)) & \text{otherwise} \end{cases} 
\end{equation*}

Clearly, there must exist a fixed point that is reached in no more than $|E|+1$ applications of $\F$.  Hence, we shall use the following notation:

\begin{equation*}
 \Fst(G) = \F^{i}(G) \textit{ s.t. } \F^{i}(G)=\F^{i+1}(G)
\end{equation*}

We now consider two agents: an attacker  and a defender.  The attacker's strategy is to destroy nodes (and their adjacent edges) in an effort to cause a cascading failure that maximizes the number of load nodes ($\Vld$) that are disconnected from all source nodes ($\Vsrc$).  Meanwhile, the defender's strategy is to harden certain nodes such that the attacker is unable to destroy them - though these nodes can be taken offline as a result of the cascading failure\footnote{Note that this would likely be the case where the attack and defense occurs in cyber-space, while the cascade occurs in the physical world.}.  The attacker can destroy $\ka$ nodes while the defender can harden $\kd$ nodes. Thus the strategy space of both the attacker and defender consists of all subsets $\Va, \Vd \subseteq V$ of size $|\Va| \leq \ka$ ($|\Vd|\leq\kd$ respectively). We denote these strategy spaces by $ATK$ ($DEF$ respectively), i.e., if we allow the attacker to consider all strategies of size $\ka$ or less we have:
\begin{equation*}
ATK = \{S \in 2^V : |S| \leq k_a \} 
\end{equation*}

We now have all of the components to define the payoff function.

\begin{definition}[Payoff Function]
Given initial network $G=(V,E)$ with edge capacities $c_{ij}(G)$, attack (defend) strategy $\Va (\Vd)$, the payoff function is defined by
\begin{equation*}
p_G(\Va,\Vd)=\disc(\Fst((V-(\Va-\Vd),E)).
\end{equation*}
\end{definition}

Now, in reality, the defender will have real-world limitations on the number of nodes (i.e. substations) he may harden.  For instance, with regard to smart grid defense, applying the most up-to-date patches on all systems may not be realistic as it could potentially require system down-time - affecting customer service.  Further, it would also likely not make sense for the defender to only harden certain nodes and ignore others.  Hence, it is reasonable to consider a situation where the defender can only harden certain nodes against attack (and may do so probabilistically - i.e. applying hardware or software updates according to a schedule).  Therefore, we study mixed strategies.  Such strategies will be specified by probability distributions  $\Pra,\Prd$  for the attacker and defender respectively.  We shall denote the number of strategies assigned a non-zero probability as $|\Pra|,|\Prd|$. We can define expected payoff as follows.

\begin{definition}[Expected Payoff]
Let $\Pra,\Prd$ be probability distributions over all subsets of $V$ of sizes $\ka$ (resp. $\kd$) or less.  These probability distributions correspond to a mixed strategy for the attacker and defender respectively.  Hence, given such probability distributions, the expected payoff can be computed as follows:
\begin{equation*}
\EP(\Pra,\Prd)=\sum_{\Va \in 2^V}\Pra(\Va)\sum_{\Vd \in 2^V}\Prd(\Vd)p_G(\Va,\Vd)
\end{equation*}
\end{definition}

In this work our goal is to find the \textit{minimax} strategy for the defender - that is the mixed strategy for the defender that minimizes the attacker's maximum expected payoff - as well as deterministic ``best responses'' for both players given the other's strategy.
\section{Computational Complexity}
\label{ccSec}
\label{negRes}
In this section, we analyze the computational complexity of determining the best response for each of the agents to a strategy of its opponent.  First, we shall discuss the case for finding a deterministic strategy for the defender and attacker.  Then we shall explore the computational complexity of finding a mixed strategy.  We summarize our complexity results in Table 3.

\begin{small}
\begin{table}
\begin{tabular}{|c|c|c|}
\hline
Opponent Strategy & Attacker & Defender \\
\hline
\hline
Mixed w. $1$ resource & NP-Compl. & PTIME \\
&Thm.~\ref{atkBrNph}&Prop.~\ref{ptime1}\\\hline
Det. w. fewer resources &NP-Compl. & PTIME \\
&Thm.~\ref{atkBrNph}&Prop.~\ref{detFewerAtk}\\\hline
Det. w. greater resources &NP-Compl. & NP-Compl.\\
&Thm.~\ref{atkBrNph}&Thm.~\ref{mainNph}\\\hline
Mixed w. fewer resources &NP-Compl. & NP-Compl.\\
&Thm.~\ref{atkBrNph}&Thm.~\ref{lessNpc}\\\hline
Mixed w. greater resources &NP-Compl. & NP-Compl.\\
&Thm.~\ref{atkBrNph}&Thm.~\ref{mainNph}\\\hline
\end{tabular}
\label{complTable}
\caption{Complexity Results for Finding a Deterministic Best Response}
\end{table}
\end{small}

We frame the formal combinatorial problem of finding the best-response for the defender as follows:\\

\noindent\textbf{Grid-Defend Deterministic Best Response (GD-DBR)}\\
\noindent INPUT: Network $G=(V,E)$, attacker mixed strategy $\Pra$ (where each option is of size no greater than $\ka$), natural number $\kd$, real numbers $X,\alpha$\\
\noindent OUTPUT: ``Yes'' if there exists a set $\Vd \subseteq V$ s.t. $|\Vd| \leq \kd$ and $\sum_{\Va \in ATK}\Pra(\Va)p_G(\Va,\Vd) \leq X$ and ``no'' otherwise.\\

We shall study this case under several conditions.  The first, and easiest case is when $\Pra = 1$ (the attacker uses a deterministic strategy) and $\ka \leq \kd$.

\begin{proposition}
\label{detFewerAtk}
When $\ka \leq \kd$ and $|\Pra| = 1$ then GD-DBR is solvable in polynomial time.
\end{proposition}
\begin{proof}
As the attacker plays only one strategy and the defender can defend at least as many nodes as are being attacked, the defender simply defends all the nodes in the attacker's strategy.
\end{proof}

However, even with $|\Pra| =1$, the problem becomes NP-hard in the case where $\ka > \kd$.

\newtheorem{theorem}{Theorem}
\begin{theorem}
\label{mainNph}
When $\ka > \kd$ then GD-DBR is NP-complete, even when $|\Pra|=1$ and $X$ is an integer.
\end{theorem}

\begin{proof}
Clearly, checking if a given deterministic defender strategy $\Vd$ meets the requirements of the ``output'' of GD-DBR can be completed in polynomial-time, providing membership in the class NP.\\

For NP-hardness consider the known NP-hard ``set cover'' problem~\cite{Garey79} that takes as input a natural number $k$, set of elements $S = \{s_1,\ldots,s_n\}$, family of subsets of $S$, $H=\{h_1,\ldots,h_m\}$ and returns ``yes'' if there is a $k$-sized (or smaller) subset of $H$ s.t. their union is equal to $S$.  We can embed Set Cover into an instance of GD-DBR in polynomial time with the following embedding: set $\ka = |H|$, $\kd = k$, $X = 0$, $\alpha = |H|+|S|$, create $G=(V,E)$ as follows:
\begin{small}
\begin{itemize}
\item For each $h \in H$ create a node $v_h$ and for each $s \in S$ create node $v_s$
\item If $s \in h$, create edge $(v_h,v_s)$, for each $ij \in E$
\item Set $\Vsrc = \{v_h | h \in H\}$, $\Vld = \{v_s | s \in S \}$, $\Va = V-\Vld$
\end{itemize}
\end{small}

Suppose, by way of contradiction (BWOC), that there is a ``yes'' answer to Set Cover but a ``no'' answer to GD-DBR.  Consider set $H'$ a subset of $H$ that is the certificate for Set Cover and the corresponding set $V'=\{v_h | h \in H'\}$ in the instance of GD-DBR.  Suppose the defender utilizes this as a strategy.  The attacker then effectively attacks the set $V-\Vld-V'$.  Note that as the graph is bi-bipartite, this does not cause any cascading failure.  By the construction, each load node must be connected to a source node, hence the number of offline load nodes is $X$.  This gives us a contradiction.

Suppose, BWOC, that there is a ``yes'' answer to GD-DBR but a ``no'' answer to the corresponding instance of Set Cover.  Let $V'$ be the certificate for GD-DBR.  We note that any element of $\Vld \cap V'$ in $V'$ can be replaced by a neighboring node from $\Vsrc$ without changing the size of this set and that such a set would still allow for all load nodes to remain online, let $V''$ be this new set.  Consider the set $\{ h | v_h \in V'' \}$.  By the contra-positive of the claim, this cannot be a cover of all elements of $S$.  However, this would also imply that there is some element $v_s \in \Vld$ that is not connected to $V''$ meaning that it fails (as the attacker successfully destroys all its neighbors).  This means that the adversary has a payoff greater than $0$ (which is what $X$ was set to) -- hence a contradiction.
\end{proof}

Hence, the presence of a more advantageous attacker is a source of complexity.  The next question would be if the attacker's behavior, i.e. deterministic vs. non-deterministic, also affects the complexity of the problem, even if the defender has the advantage.  First, let us examine the case where the attacker has a mixed strategy with $\ka=1$.

\begin{proposition}
\label{ptime1}
When $\ka =1$ then GD-DBR is solvable in polynomial time (w.r.t. $|\Pra|$), even when $|\Pra|\geq 0$.
\end{proposition}
\begin{proof}

In this case, we can re-write the payoff function as $p_G(\{v\},\Vd) = 0$ if $v \in \Vd$ and $p_G(\{v\},\Vd) = p_G(\{v\},\emptyset)$ otherwise.  Let $V' = \cup\{ \Va \in ATK | \Pra(\Va) > 0 \}$.  Note that each element of $V'$ is also a strategy the attacker plays with a non-zero probability (as the attacker only plays singletons).  Hence, the expected payoff can be re-written as  $\sum_{v \in V'-\Vd}\Pra(\{v\})p_G(\{v\},\emptyset)$.  Therefore, the best a defender can do is defend the top $\kd$ nodes in $V'$ where\\ $\Pra(\{v\})p_G(\{v\},\emptyset)$ is the greatest - which can be easily computed in polynomial time and allows us to determine the answer to GD-DBR.
\end{proof}

However, if the defender is playing a mixed strategy with $\ka >1$, then the problem again becomes NP-complete.

\begin{theorem}
\label{lessNpc}
When $|\Pra|>1$ and $\ka >1$, GD-DBR is NP-complete, even when $\kd > \ka$ and $X$ is an integer.
\end{theorem}
\begin{proof}
NP-completeness mirrors that of Theorem~\ref{mainNph}.  For NP-hardness, we again consider a reduction from set-cover (defined in the proof of Theorem~\ref{mainNph}.  The embedding can again be performed in polynomial time as follows: set $\ka = \max_{s\in S}|\{h | s \in h\}|$, set $\kd = k$, $X = 0$, $\alpha = |H|+|S|$, create $G=(V,E)$, $\Vsrc$, and $\Vld$ as per the construction in Theorem~\ref{mainNph}.  We then set up the mixed strategy as follows: for each $s \in S$, let $\Va^s = \{h | s \in h\}$ and $\Pra(\Va^s)=1/|S|$.

Suppose, BWOC, that there is a ``yes'' answer to set cover and a ``no'' answer to the instance of GD-DBR.  Consider set cover solution $H^*$ and set $\Vd = \{v_h | h \in H^*\}$.  Note that $\Vd$ meets the cardinality requirement.  Note that by the construction, a source node becomes disconnected only if all of the load nodes connected to it are attacked, hence there is some node in the set $\Vld$ that is totally disconnected under at least one attacker strategy - let $v_s$ be this node.  However, as set $H^*$ covers $S$, then regardless of the attacker strategy, there is always some node $v_h$ that is connected and never attacked (giving the attacker a payoff of zero) - hence a contradiction.

Suppose, BWOC, that there is a ``yes'' answer to GD-DBR and a ``no'' answer to the instance of set cover.  Consider GD-DBR solution $V'$.  We note that any element of $\Vld \cap V'$ in $V'$ can be replaced by a neighboring node from $\Vsrc$ without changing the size of this set and that such a set would still allow for all load nodes to remain online, let $V''$ be this new set.  Consider the set $H^* = \{h | v_h \in V''\}$.  Note that $|H^*| \leq k$.  By the contra-positive, there must be at least one element of $S$ not covered by $H^*$.  Let node $v_s$ be a node associated with uncovered element $s$.  As GD-DBR returned ``yes'' then there is no attacker strategy where $v_s$ becomes disconnected from some node in $\Vsrc$.  As attack strategy $\Va^s$ includes all nodes that are connected to $v_s$, then at least one of these nodes must be included in $V''$.  Therefore, for every node $v_s \in \Vld$ there is some node $v_h \in \Vld \cap V''$ that is connected to it, which means, by the construction, that $H^*$ must cover all elements of $S$ - a contradiction.
\end{proof}

We now frame the formal problem for finding a deterministic best-response for the attacker below.\\

\noindent\textbf{Grid-Attack Deterministic Best Response (GA-DBR)}\\
\noindent INPUT: Network $G=(V,E)$, defender mixed strategy $\Prd$ (where each option is of size no greater than $\kd$), natural number $\ka$, real numbers $X,\alpha$\\
\noindent OUTPUT: ``Yes'' if there exists a set $\Va \subseteq V$ s.t. $|\Va| \leq \ka$ and $\sum_{\Vd \in DEF}\Prd(\Vd)p_G(\Va,\Vd) \geq X$ and ``no'' otherwise.\\

In the case of $\ka=1$, this problem is solvable in polynomial time: simply consider each $v \in V$.  The attacker computes $\sum_{\Vd \in DEF}\Prd(\Vd)p_G(\{v\},\Vd)$ until one is found that causes the payoff to exceed or be equal to $X$.  However, for strategies of larger size, the problem becomes NP-hard, regardless of the size of the defender strategy.

\begin{fact}
\label{easyFact}
When $\ka=1$, GA-DBR is solvable in polynomial time (w.r.t. $|\Prd|$).
\end{fact}

\begin{theorem}
\label{atkBrNph}
GA-DBR is NP-complete.
\end{theorem}
\begin{proof}
Clearly, a certificate consisting of a set $\Va \subseteq V$ can be verified in polynomial time, giving us membership in NP.  For NP-hardness consider the known NP-hard ``vertex cover'' problem~\cite{Garey79} that takes as input a graph $G'=(V',E')$ (with no self-loops) and natural number $k$ and returns ``yes'' iff there is a set of $k$ or fewer vertices that are adjacent to each edge in $E$.  We can embed vertex cover into an instance of GD-DBR in polynomial time with the following embedding: set $\ka = k$, $\kd = 0$, $\Vd = \emptyset$, $X = |V'|$, $\alpha = |E|$, $G = G'$, and $\Vsrc = \Vld = V'$.

Suppose, BWOC, the above problem instance provides a ``yes'' answer to the vertex cover problem but a ``no'' answer to GA-DBR.  Let $V''$ be a vertex cover of size $k$ or less for $G'$.  Consider the corresponding set of vertices in $G$ (we shall call this $V^*$).  Note that $|V^*|\leq \ka$.  As an attacker attacking $V^*$ disconnects those nodes from the network, all edges adjacent to $V^*$ fail.  As $V^*$ is a vertex cover for $G$, this means that there are no edges in the graph once $V^*$ is removed.  Hence, no load node is connected to any source node - giving the attacker a payoff of at least $X$ -- hence a contradiction.

Suppose, BWOC, the above problem instance provides a ``yes'' answer to GA-DBR but a ``no'' answer to the vertex cover problem.  Let $\Va$ be the set of nodes the attacker attacks in GA-DBR.  As $\alpha=|E|$ and as $\Vsrc = V$, nodes only fail in a cascade if they are either targeted by the attacker or become totally disconnected.  Further, as $X=|V|$, all nodes in $G$ are either in $\Va$ or disconnected - meaning that $\Va$ must be a vertex cover of size $\ka$ or less.  As $\ka = k$ we have a contradiction.
\end{proof}


Due to the use of covering problems for the complexity results in Theorems~\ref{mainNph}, \ref{lessNpc}, and~\ref{atkBrNph}, it may seem reasonable to frame the problem as a sub- or super- modularity optimization where the objective function is monotonic.  However, here we show (unfortunately) that these properties do not hold for either player.  First, we shall make statements regarding the monotonicity of the payoff function.

\begin{proposition}
\label{monoProp}
Iff $\forall \Vd^*$, $\Va \subseteq \Va'$: $p_G(\Va,\Vd^*) \leq p_G(\Va',\Vd^*)$ then
$\forall \Va^*$, $\Vd \subseteq \Vd'$: $p_G(\Va^*,\Vd) \geq p_G(\Va^*,\Vd')$.
\end{proposition}

The idea of \textit{submodularity} can be thought of as ``diminishing returns.''  Given a set of elements $S$ and a function $f : 2^S \rightarrow \Re^+$, we say a $f$ is submodular if for any sets $S_1 \subseteq S_2$ and element $s \notin S_2$, we have the following relationship:
\begin{eqnarray*}
f(S_1 \cup \{ s \})-F(S_1) \geq f(S_2 \cup \{ s \})-F(S_2)
\end{eqnarray*}

A complementary idea of supermodularity is also often studied - in this case the inequality is reversed.  Unfortunately, when we fix the strategy for the defender, the attacker strategy is neither submodular nor supermodular - making the dynamics of this model significantly different from others (i.e. \cite{tsai12}).  Let consider strategies $\Va,\Vd$ where $\Va$ causes some load node $v \notin (\Va\cup\Vd)\cap\Vld$ to disconnect and any node the strategy $\{v\}$ causes to disconnect will also become disconnected with strategy $\Va$ (such a case is easy to contrive, particularly with a bi-partite network).  Therefore, we get the following relationship:
\begin{eqnarray*}
p_G(\Va \cup \{v\},\Vd) - p_G(\Va,\Vd) < p_G(\{v\},\Vd) - p_G(\emptyset,\Vd)
\end{eqnarray*}
This arises from the fact that the left-hand side of the above equation becomes zero and the right hand side of the equation is equal to $p_G(\{v\},\Vd)$ which must be at least one.  Now consider another example.  Suppose we have a simple V-shaped network of three nodes.  The angle of the V is a load node, while the other two nodes are source nodes.  With $\alpha=1$, the load node receives power if at least one of the source nodes is connected to it.  However, it does not require both.  Let $\Va$ be a strategy consisting of one source node and $v$ be the other source node, and $\Vd$ consist of the load node.  From this, we have the following relationship:
\begin{eqnarray*}
p_G(\Va \cup \{v\},\Vd) - p_G(\Va,\Vd) > p_G(\{v\},\Vd) - p_G(\emptyset,\Vd)
\end{eqnarray*}
In this case, the right-hand side becomes zero while the left hand side becomes one.  This leads us to the following fact:

\begin{fact}
\label{nonSmFact}
When $\Vd$ is fixed, $p_G$ is neither submodular nor supermodular.
\end{fact}

Now let us consider when we fix the attacker's strategy.  If the payoff is submodular when the attacker's strategy is fixed, then we have the following for $\Vd \subseteq \Vd'$ and $v \notin \Vd'$ if the payoff subtracted from the number of nodes is submodular:
\begin{eqnarray*}
p_G(\Va, \Vd' \cup \{v\})-p_G(\Va, \Vd') \geq p_G(\Va, \Vd \cup \{v\})-p_G(\Va, \Vd)
\end{eqnarray*}
This is equivalent to the following:
\begin{eqnarray*}
&p_G(\Va-(\Vd' \cup \{v\}),\emptyset)-p_G(\Va-\Vd',\emptyset) \geq &\\
&\,\,\,\,\,\,\,\,\,\,\,\,\,\,\,\,\,\,\,\,\,\,\,\,\,\,\,p_G(\Va- (\Vd \cup \{v\}),\emptyset)-p_G(\Va- \Vd,\emptyset)&
\end{eqnarray*}
Now let $\Va' = \Va-(\Vd' \cup \{v\})$ and $\Va'' = \Va' \cup (\Vd'-\Vd)$.  Clearly $\Va'' \supseteq \Va'$ and $v \notin \Va''$.  Now we get the following:
\begin{eqnarray*}
p_G(\Va',\emptyset)-p_G(\Va'\cup\{v\},\emptyset) &\geq& p_G(\Va'',\emptyset)-p_G(\Va''\cup \{v\},\emptyset)\\
p_G(\Va'\cup\{v\},\emptyset) -p_G(\Va',\emptyset)&\leq& p_G(\Va''\cup \{v\},\emptyset)-p_G(\Va'',\emptyset)
\end{eqnarray*}
Hence, submodualrity of the payoff function when the attacker's strategy is fixed would give us supermodualrity of the payoff function when the defender's strategy is fixed at the empty set.  However, this clearly violates Fact~\ref{nonSmFact} and gives rise to the following:

\begin{fact}
\label{nonSmFact2}
When $\Va$ is fixed, $p_G$ is neither submodular nor supermodular.
\end{fact}
\section{Algorithms}
\label{analySec}
In this section, we present heuristic algorithms for finding the deterministic best response of each player as the results of the previous section generally preclude a polynomial time algorithm for an exact solution.  We first introduce a version of a ``high load'' strategy for the defender based on the ideas of \cite{crucitti04}.  Then we introduce a greedy heuristic for each player.  This is followed by our approach for finding mixed strategies based on the double-oracle algorithm of \cite{mcmahan03}.\\

\noindent\textbf{Hi-Load Node Approach.}  In \cite{crucitti04}, the authors study ``high load'' nodes: nodes through which the greatest number of shortest paths pass.  They show that attacks on these nodes tend to initiate cascading failures -- suggesting that they should be a priority for defense.  We formalize the definition of nodal load in our framework (essentially an extended definition of node betweenness~\cite{wasserman1994social}) by extending our function $\sp$ for nodes as follows.

\begin{definition}[Nodal Load]
For a given node, the \textbf{nodal load} is defined as the sum of the fraction of shortest paths for each pair that pass through that node.  Formally:
\begin{equation*}
\sp(i) = \sum_{s \in \Vsrc, t \in \Vld}\frac{\sigma_G(s,t | i)}{\sigma_G(s,t)},
\end{equation*}
where $\sigma_G(s,t | i)$ is the number of shortest paths between $s,t \in V$ that pass through node $i$.
\end{definition}

Hence, we shall refer to the \textit{Deterministic Load-Based} or DLB strategy for the defender as one in which he deterministically protects the $\kd$ nodes with the greatest load.  We note that this is not necessarily a ``best response'' but the intuition is that defense will occur at nodes that are perceived to be critical to the adversary.  This intuition is similar to that of the ``most vital arc'' idea seen in other failure model games~\cite{Alderson:2013-03-01T00:00:00:1082-5983:21,salmeron_04}.\\

\noindent\textbf{Greedy Heuristics for Finding Deterministic Strategies.}  Here we present a simple greedy heuristic to find the defender's best-response (\textsf{GREEDY\_DEFENDER\_RESP}).  The analogous heuristic for the attacker is not shown due to space constraints, but we shall refer to it as\\ \textsf{GREEDY\_ATTACKER\_RESP}.  We note that while we do not make general approximation guarantees (due to the results in Section~\ref{negRes}), we note that by Proposition~\ref{monoProp}, that nodes added in step~\ref{addNode} will always cause an increase in payoff to the defender (and in the analogous greedy approach for the attacker, this holds true as well).  Further, by Proposition~\ref{ptime1}, when $\ka=1$, we can be sure that \textsf{GREEDY\_DEFENDER\_RESP} returns an exact solution, even when the attacker has a mixed strategy.  Unfortunately, by Theorem~\ref{atkBrNph}, the same cannot be said if the greedy heuristic is used for the attacker's best response.\\

\algsetup{indent=1em}
	\begin{algorithm}[!h]
		\caption{ \textsf{GREEDY\_DEFENDER\_RESP}}
		\begin{algorithmic}[1]
\begin{small}
		\REQUIRE Mixed strategy $\Pra$, Natural number $\kd$
		\ENSURE Set of nodes $ \Vd $
		\medskip
		\STATE{ $\Vd = \emptyset$}
		\STATE{ Let $ATK$ be the set of strategies associated with $\Pra$}
		\STATE{ Set $flag = \textsf{True}$, $p^* = -\infty$}
		\WHILE{ $|\Vd| \leq \kd$ and $flag$ and $p^* <0$}
			\STATE{ $p^* = -\sum_{\Va \in ATK}\Prd(\Va)p_G(\Va,\Vd)$}
			\STATE{ $curBest = null$, $curBestScore=0$, $haveValidScore=\textsf{False}$ }
			\FOR{ $i \in V-\Vd$ }
				\STATE{ $curScore = p^*-\sum_{\Va \in ATK}\Prd(\Va)p_G(\Va,\Vd\cup\{i\})$}
				\IF{$curScore \geq curBestScore$}
					\STATE{ $curBest = i$}
					\STATE{ $curBestScore = curScore$}
					\STATE{ $haveValidScore=\textsf{True}$}
				\ENDIF
			\ENDFOR
			\IF{ $haveValidScore = \textsf{False}$}
				\STATE{ $flag=\textsf{False}$}
			\ELSE
				\STATE\label{addNode}{ $\Vd = \Vd \cup \{curBest \}$}
			\ENDIF
		\ENDWHILE
		\RETURN{ $ \Vd $}.
\end{small}
	\end{algorithmic}
\end{algorithm}

\noindent\textbf{Finding Mixed Strategies.}  If the attacker uses a mixed strategy that consists of uniformly attacking elements of $\{S\subset V_{ld}: |S|=k_a\}$ then the best any pure defender strategy can do is defending $V_{d} \subset V_{ld}$. The attacker's strategy implies that any node in $V_{ld}$ is attacked with probability $\frac{ k_a }{ |V_{ld}| }$. Each of the $|V_{ld}|-k_a$ remaining nodes in $V_{ld}$ is then disconnected with probability $\frac{ k_a }{ |V_{ld}| }$, i.e., $x \geq k_a(1-\frac{k_d}{|V_{ld}|})$. Clearly due to the cascading the value of the game will probably be higher, illustrating the disadvantage the defender has in this game. To determine both player's optimal strategies and the value of the game we resort to an algorithmic approach. We find the defender's optimal strategy with the following linear program.  We can find minimax strategy for the defender with the following linear program.  It simply assigns a probability to each of the defenders strategies in a manner that minimizes the maximum payoff for the adversary.  As a consequence, the solution to the following linear program, \textsf{DEF\_LP} can provide the mixed minimax strategy for the defender.  An analogous linear program, \textsf{ATK\_LP} (not shown), which mirrors \textsf{DEF\_LP}, will provide that result for the attacker.

\begin{definition}[\textsf{DEF\_LP}]
\begin{small}
\begin{eqnarray}
&\min p^*&\\
subj. to& p^* \geq \sum_{\Vd \in DEF}X_{\Vd}p_G(\Va,\Vd) & \forall \Va \in ATK\\
& 1=\sum_{\Vd \in DEF}X_{\Vd}&\\
& X_{\Vd}\in [0,1]&\forall \Vd \in DEF
\end{eqnarray}
\end{small}
\end{definition}

Note that the above linear program requires one variable for each of the defender's strategies and one constraint for each of the attacker's strategies.  However, as there are a combinatorial number of strategies, even writing down such a linear program is not practical except for very small problem instances.  To address this issue of intractability, we employ the double-oracle framework for zero-sum games introduced in \cite{mcmahan03} and has been applied in more recent work as well~\cite{DBLP:conf/atal/BosanskyKLCP13,jain13a}.  We present the algorithm \textsf{DOUBLE\_ORACLE} as follows:

\algsetup{indent=1em}
	\begin{algorithm}[!h]
		\caption{ \textsf{DOUBLE\_ORACLE}}
		\begin{algorithmic}[1]
\begin{small}
		\REQUIRE Network $G=(V,E)$, natural number $maxIters$
		\ENSURE Mixed defender strategy $\Prd$
		\medskip
		\STATE Initialize $numIters = 0$, $flag = \textsf{True}$
		\STATE{Initialize the sets of strategies $ATK,DEF$ to both be $\{ \emptyset \}$ }
		\WHILE{$flag$ and $numIters \leq maxIters$}
			\STATE\label{lpStep}{ Create $\Pra,\Prd$ based on the solutions to $\textsf{ATK\_LP}$ and $\textsf{DEF\_LP}$ respectively.}
			\STATE\label{respStep}{\textbf{IF} $numIters<maxIters$ \textbf{THEN} let $\Va$ be the attacker's best response to $\Prd$ and $\Vd$ be the defender's best response to $\Pra$}
			\STATE{\textbf{IF} $\Va \in ATK$ and $\Vd \in DEF$ \textbf{THEN} $flag = \textsf{False}$ \textbf{ELSE} $ATK=ATK\cup\{\Va\}$, $DEF=DEF\cup\{\Vd\}$}
			\STATE{ $numIters += 1$}
		\ENDWHILE
		\RETURN{ $ \Pra $}.
\end{small}
	\end{algorithmic}
\end{algorithm}

The intuition behind the above algorithm is that it iteratively creates mixed strategies for both the attacker and defender based on a solution to a linear program over the sets of current possible strategies for both players ($ATK,DEF$).  This is followed by finding (for each player) the best deterministic response to it's opponent's strategy.  If these new strategies are both already in the set of possible strategies for the respective players, the algorithm terminates.  Otherwise, they are added to $ATK,DEF$ respectively.  We note that by Theorem 1 of \cite{mcmahan03} that the above algorithm will guarantee an exact solution if $maxIters$ is set to the number of possible strategies.  In practice, \cite{mcmahan03} demonstrates that the algorithm converges much faster.

In \textsf{DOUBLE\_ORACLE}, the finding the solutions to \textsf{DEF\_LP}, \textsf{ATK\_LP} will be tractable provided that the algorithm converges in a polynomial number of steps (either through convergence or after the specified $maxIters$).  However, as we have shown, computing the best responses is usually computationally difficult.  Although, we note in the case where $\ka=1$, that by Proposition~\ref{ptime1} and Fact~\ref{easyFact}, the double oracle algorithm will return an optimal solution, even if greedy approximations are used for the oracles (provided it runs until convergence).
\section{Experimental Evaluation}
\label{expSec}

All experiments were run on a computer equipped with an Intel X5677 Xeon Processor operating at 3.46 GHz with a 12 MB Cache and 288 GB of physical memory.  The machine was running Red Hat Enterprise Linux version 6.1.  Only one core was used for experiments.  All algorithms were coded using Python 2.7 and leveraged the NetworkX library\footnote{http://networkx.lanl.gov/} as well as the PuLP library for linear programming\footnote{http://pythonhosted.org/PuLP/}. All statistics presented in this section were calculated using the R statistics software.

In our experiments, we utilized a dataset of an Italian 380 kV power transmission grid~\cite{rosato08}.  This power grid network consisted of $310$ nodes of which $113$ were source, $96$ were load, and the remainder were transmission nodes.  The nodes were connected with $361$ edges representing the power lines.

In our initial experiments, we examined the properties of the model when no defense is employed.  In Figure~\ref{bet_vs_payoff} (left) we show results concerning nodal load vs. the payoff achieved by the adversary if that node is attacked (and no others).  Interestingly, we noticed a significant number of nodes with low nodal load yet high-payoff if attacked (see nodes in dashed box).  This may suggest that the DLB strategy may be insufficient in some cases.  Later we see how DLB fails to provide adequate in a defense against the attacker best response to DLB.  This is likely due to these hi-payoff, low-load nodes.  In Figure~\ref{bet_vs_payoff} (right) we examine $\alpha$ (capacity margin) vs. attacker payoff for various settings of $\ka$ (using the \textsf{GREEDY\_ATTACKER\_RESP} heuristic).  Here we found that, in general, payoff decreases linearly with capacity margin ($R^2 \geq 0.84$ for each trial).

\begin{figure}
    \begin{center}
        \includegraphics[width=.85\linewidth]{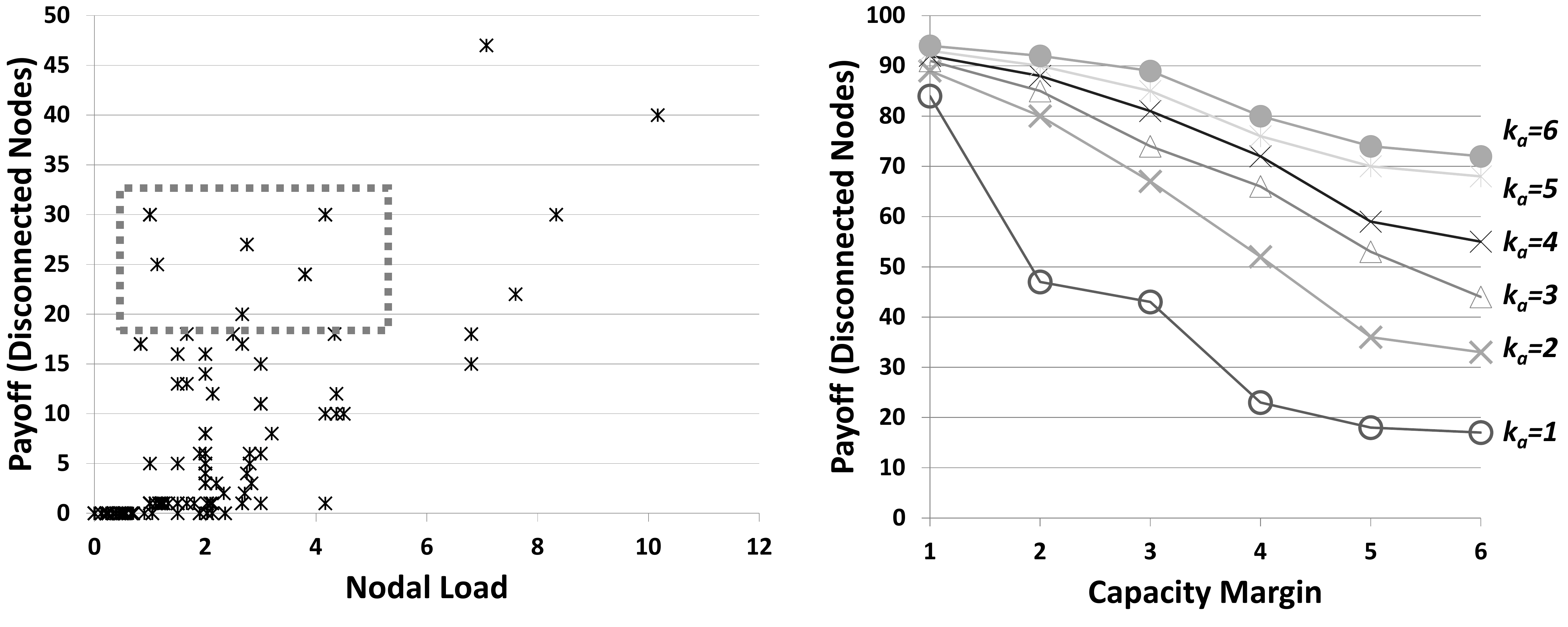}
    \end{center}
    \caption{\textbf{Left:} Nodal load vs. payoff (note hi-payoff, low-load nodes in the dashed box), \textbf{Right:} Capacity margin ($\alpha$) vs. payoff}
    \label{bet_vs_payoff}
\end{figure}

\begin{figure}[!h]
    \begin{center}
        \includegraphics[width=.85\linewidth]{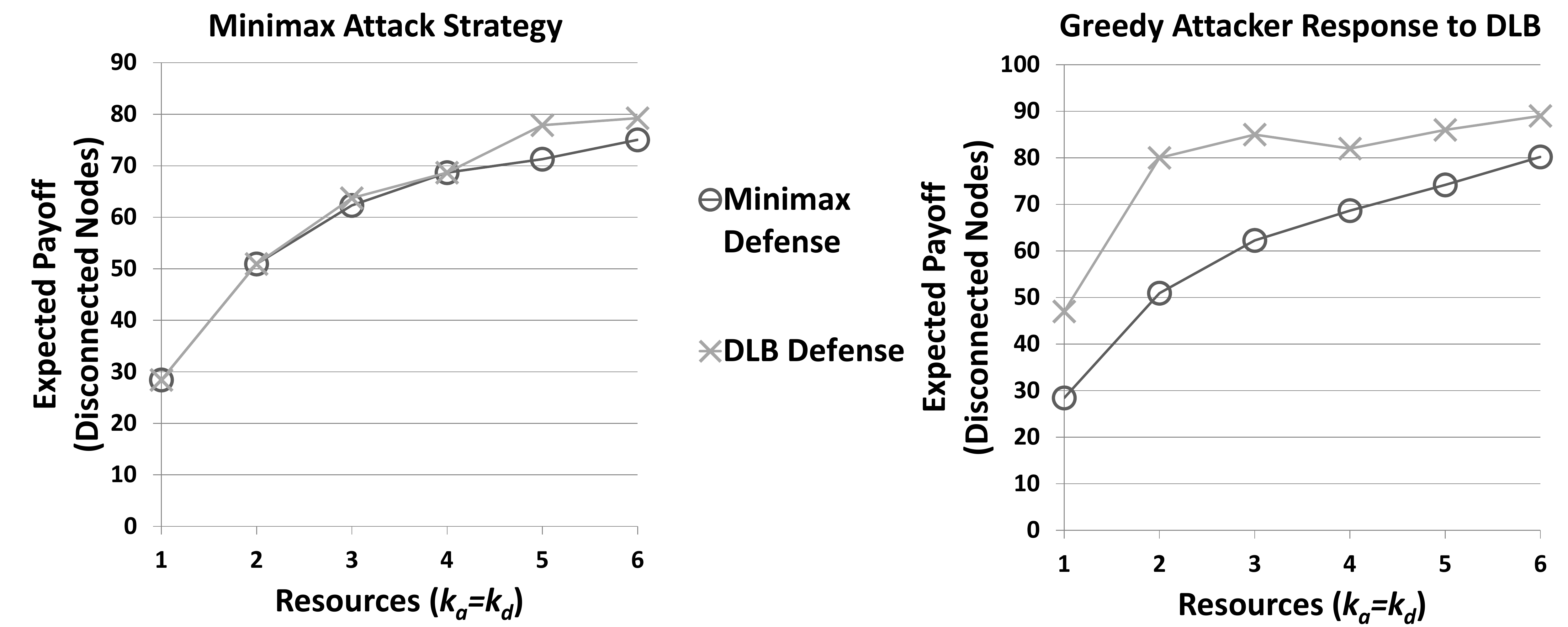}
				\includegraphics[width=.85\linewidth]{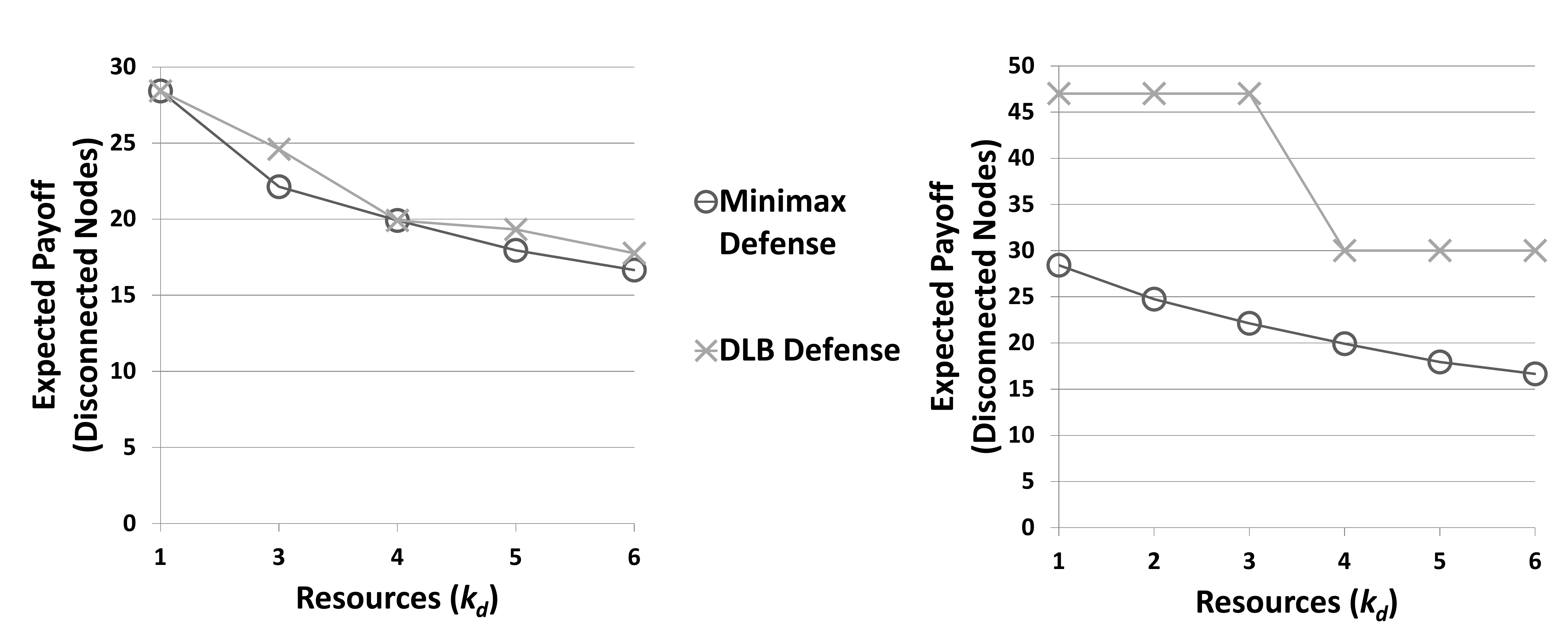}
				\includegraphics[width=.85\linewidth]{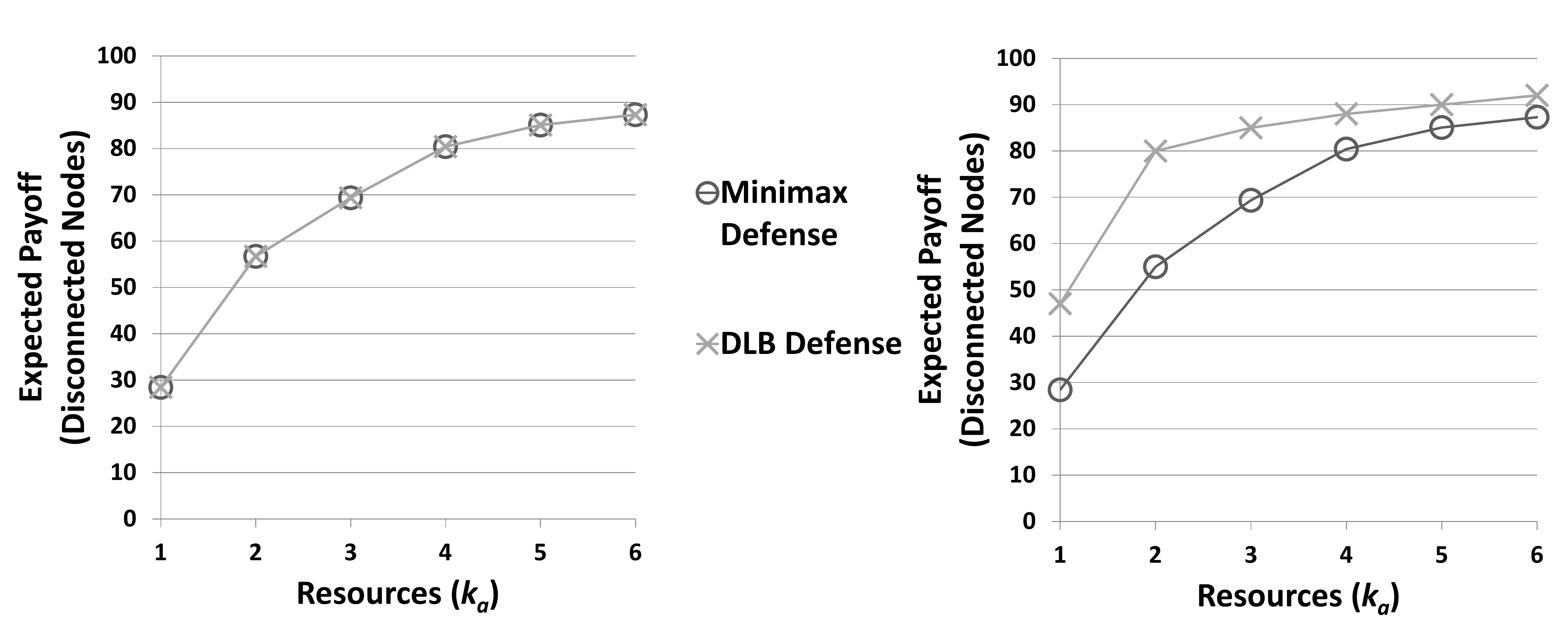}
				\end{center}
    \caption{Minimax and DLB defense strategies vs. minimax attack strategy (left) and the attacker's greedy best response to DLB (right).  Examined are the cases where $\ka=\kd$ (top), $\ka=1$, $\kd$ varies (middle) and $\kd=1$, $\ka$ varies (bottom).}
    \label{ka1fig}
\end{figure}

Next, we examined the relative performance of the minimax (mixed) defense strategy and the DLB strategy under different resource constraints and against the minimax (mixed) attack strategy as well as the attacker's (deterministic) greedy response to the DLB defense.  In these experiments, we considered the case where both players have equal resources, the attacker has one resource (which by Proposition~\ref{ptime1} and Fact~\ref{easyFact} we are guaranteed an optimal solution), and the defender has one resource.  These results are displayed in Figure~\ref{ka1fig}.  In these trials we set the capacity margin $\alpha=0.5$, meaning that all edges had an excess capacity of $50\%$.  We did not use the $maxIters$ parameter of the \textsf{DOUBLE\_ORACLE} algorithm, but instead allowed it to run until convergence.

With regard to the comparison between DLB and minimax defense, both performed comparably against the minimax attack strategy.  In fact, an analysis of variance (ANOVA) indicated little variance between the two when faced with the minimax attacker ($p \geq 0.74$ for these trials).  Yet, a defender known to be playing a single strategy would likely not face an attacker who plays the minimax strategy, but rather the best response to the DLB.  In this case, DLB play resulted in significantly greater payoff to the attacker than the defender ($p \leq 0.29$ for these trials, the DLB defense results in $15.6$ more disconnected nodes on average).  This failure of the DLB strategy to perform well against a deterministic attacker best response is likely due to  the presence of low-load yet high-payoff nodes as shown in Figure~\ref{bet_vs_payoff}.

We also noticed that an increase in resources seems to favor the attacker more than the defender.  When both players played their respective minimax strategy, the expected payoff for the attacker increased monotonically with the cardinality of the strategies.  Further, when $\kd=1$ and $\ka$ was greater, the attacker's payoff tripled when his resources increased from $1$ to $6$.  However, when $\ka=1$ and $\kd$ was greater, the defender's payoff only increased by a factor of $1.7$.  Hence, the attacker can cause more damage than the defender can mitigate with the same amount of extra resources.  We suspect that this is likely because a defended node can still fail during a cascade - which would likely be the case if the attack and defense operations are restricted to cyber-space, where physical system failure may still be possible as the result of a cascade initiated by virtual means.

We also examined the run-time of our approach, as displayed in Figure~\ref{ds_fig} (left).  Though run-time did seem to scale linearly with strategy size ($R^2 = 0.90 \pm 0.2$ for each experiment), it appears that run-time will in general prohibit the study of larger strategies or networks (our longest experiment ran for 12 days).  In examining the iterations of the \textsf{DOUBLE\_ORACLE} algorithm, Figure~\ref{ds_fig} (left), we find that run-time of an iteration of the algorithm progressively increases (note that this figure is showing the run-time for each iteration, not a cumulative time).  This increase is likely the combined result of the growing linear program and the growing size of the mixed strategies considered by the greedy approximation sub-routines.  We are currently exploring reliable methods to limit the number of iterations while maintaining defender payoff.

\begin{figure}
    \begin{center}
        \includegraphics[width=.85\linewidth]{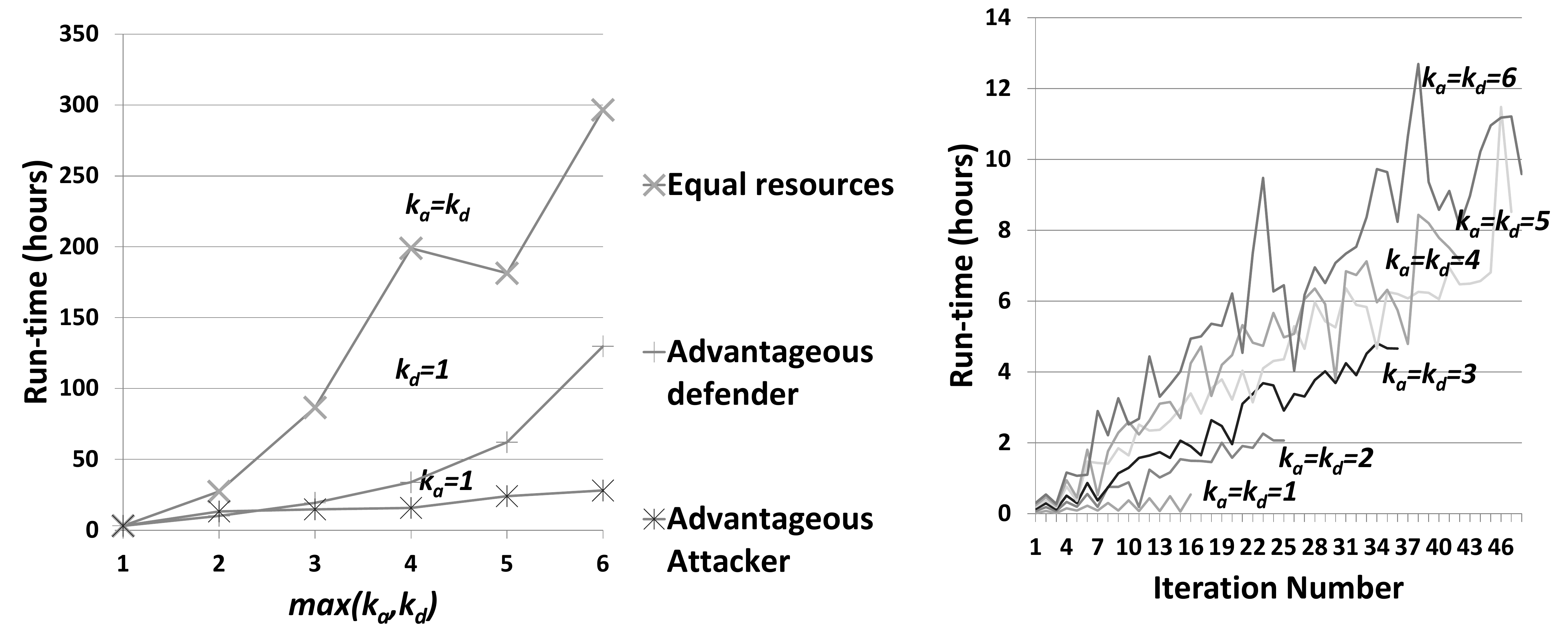}
    \end{center}
    \caption{Strategy size vs. run-time in hours (left) and the run-time of each iteration for the experiments where $\ka=\kd$}
    \label{ds_fig}
\end{figure}

\section{Related Work}
\label{rwSec}
Network security has received much attention from the research community in the past two decades.  Recent incidents have shown that due to their internet connectedness such networks can come under cyber attack, causing severe problems\footnote{http://www.wired.com/threatlevel/2009/10/smartgrid/}.  See \cite{wei} for a discussion of cyber-security issues relevant to smart grid grids.
The utilization of game theory in designing defense solutions seems ubiquitous. For instance \cite{Liu:2005:IMI:1053283.1053288} model the interaction between a DDoS attacker and the network administrator while \cite{Liu:2006:BGA:1190195.1190198} considers a game theoretic formulation for intrusion detection. Other formulations consist include stochastic games~\cite{conf/icc/NguyenAB09}, signaling games~\cite{1437828}, allocation games~\cite{bloem06intrusion} and repeated games~\cite{1430267}. Game theory is also being used in monitoring and decision making in smart grids, see for instance \cite{DBLP:journals/tsg/EsmalifalakSHS13} or the survey by Fadlullah et al.~\cite{6096962}.   However to date no game theoretic approach has been given for the specific problem where the attacker explicitly sets of a cascading power failure to maximize the damage to the defender.

Cascading failure models applied to power grid infrastructure have been studied in the past~\cite{buld10,crucitti04,motter02}.  The model of \cite{crucitti04} introduces the idea of edge failure based on excessive loads.  The goal of the research presented in these papers was to illustrate properties of the cascade, rather than explore strategies for attack and defense as this work does.  There has been work on attack and defense of a power-grid network under the DC power-flow mode~\cite{Alderson:2013-03-01T00:00:00:1082-5983:21,salmeron_04,rosato08,Brown:2006:DCI:1235123.1235128}.  However, the DC power flow model is not designed to model the more rapid cascading failures (i.e. the 2003 cascading failure in the eastern United States~\cite{ohFailure}).

The application of game theory to security situations was made popular by \cite{Paruchuri:2008:PGS:1402298.1402348} where it used for airport security patrol scheduling.  Since then, other applications have emerged including port protection~\cite{Shieh:2012:PDG:2343576.2343578}, finding weapons caches~\cite{Shakarian:2012:AGA:2089094.2089110}, and security checkpoint placement~\cite{jain13a}.  One that bears similarity to this work is \cite{tsai12} - studying games for controlling contagions on a network.  However, as previously discussed, that model operates under very different dynamics.

\pagebreak
\section{Conclusion}
In this paper, we explored complexity, algorithmic, and implementation issues in a two-player security game where the attacker/defender look to create/mitigate cascading failure on a power grid.  Future work includes an examination of scalability issues (larger networks and strategies), adding uncertainty to the model, and the consideration of more real-world information about the power grid network (i.e. actual line capacities, etc.) in order to create a richer model.
\break
\section{Acknowledgments}
We would like to thank D. Alderson for his input on related work and V. Rosato for providing us the power grid dataset.  Some of the authors are supported by ARO project 2GDATXR042.
The opinions in this paper are those of the authors and do not necessarily reflect the opinions of the funders, the U.S. Military Academy, or the U.S. Army.

\bibliographystyle{abbrv}
\bibliography{network}  

\end{document}